\documentclass[a4paper]{article}
\usepackage[margin=30mm]{geometry}
\usepackage{graphicx}
\usepackage{wrapfig}
\usepackage{authblk}

\usepackage{lineno,hyperref}

\usepackage{amsmath,amssymb}
\usepackage{amsthm}
\usepackage{graphics}
\usepackage[procnumbered,ruled,vlined,linesnumbered]{algorithm2e}
\usepackage{tikz}

\newtheorem{theorem}{Theorem}
\newtheorem{lemma}[theorem]{Lemma}

\newcommand{\I}{\mathcal{I}}
\renewcommand{\L}{\mathtt{L}}
\newcommand{\R}{\mathtt{R}}
\renewcommand{\P}{\mathcal{P}}
\newcommand{\mrm}[1]{\mathrm{#1}}
\newcommand{\mtt}[1]{\mathtt{#1}}

\newcommand{\floor}[1]{\lfloor #1 \rfloor}

\begin{document}
\title{Efficient Non-isomorphic Graph Enumeration Algorithms for Subclasses of Perfect Graphs}
\author[1]{Jun Kawahara}
\author[2]{Toshiki Saitoh}
\author[2]{Hirokazu Takeda}
\author[3]{Ryo Yoshinaka}
\author[2]{Yui Yoshioka}
\affil[1]{Kyoto University}
\affil[2]{Kyushu Institute of Technology}
\affil[3]{Tohoku University}
\date{}

\maketitle              
\begin{abstract}
Intersection graphs are well-studied in the area of graph algorithms. 
Some intersection graph classes are known to have algorithms enumerating all unlabeled graphs by reverse search. 
Since these algorithms output graphs one by one and the numbers of graphs in these classes are vast, they work only for a small number of vertices. 
Binary decision diagrams (BDDs) are compact data structures for various types of data and useful for solving optimization and enumeration problems. 
This study proposes enumeration algorithms for five intersection graph classes, which admit $\mathrm{O}(n)$-bit string representations for their member graphs. 
Our algorithm for each class enumerates all unlabeled graphs with $n$ vertices over BDDs representing the binary strings in time polynomial in $n$. 
Moreover, our algorithms are extended to enumerate those with constraints on the maximum (bi)clique size and/or the number of edges.

\end{abstract}

\section{Introduction}
This paper is concerned with efficient enumeration of unlabeled intersection graphs.
An intersection graph has a geometric representation such that each vertex of the graph corresponds to a geometric object and the intersection of two objects represents an edge between the two vertices in the graph. 
Intersection graphs are well-studied for their practical and theoretical applications~\cite{Brandstadt:1999,Spinrad}. 
For example, interval graphs, which are represented by intervals on a real line, are applied in bioinformatics, scheduling, and so on~\cite{Golumbic04}.
Proper interval graphs are a subclass of interval graphs with interval representations where no interval is properly contained to another. 
These graph classes are related to important graph parameters: The bandwidth of a graph $G$ is equal to
 the smallest value of the maximum clique sizes in proper interval graphs that extend $G$~\cite{KaplanS96}.

The literature has considered the enumeration problems for many of the intersection graph classes.
The graph enumeration problem is to enumerate all the graphs with $n$ vertices in a specified graph class. 
If it requires not enumerating two isomorphic graphs, it is called \emph{unlabeled}. Otherwise, it is called \emph{labeled}. 
Unlabeled enumeration algorithms based on reverse search~\cite{AvisF96} have been proposed for subclasses of interval graphs and permutation graphs~\cite{SOYU12,SYKU10,YamazakiQU21,YamazakiSKU20}.
Those algorithms generate graphs in time polynomial in the number of vertices per graph. 
In this regard, those algorithms are considered to be fast in theory. 
However, since those algorithms output graphs one by one and the numbers of graphs in these classes are vast, the total running time will be impractically long, and storing the output graphs requires a large amount of space. 

The idea of using \emph{binary decision diagrams (BDDs)} has been studied to overcome the difficulty of the high complexity of enumeration.
BDDs can be seen as indexing and compressed data structures for various types of data, including graphs, via reasonable encodings.  
The technique so-called \emph{frontier-based search}, given an arbitrary graph, efficiently constructs a BDD which represents all subgraphs satisfying a specific property~\cite{KawaharaIIM17,Knuth2014art,SekineIT95}. Among those, Kawahara et al.~\cite{KawaharaSSY19} proposed enumeration algorithms for several sorts of intersection graphs, e.g., chordal and interval graphs.
Using the obtained BDD, one can easily count the number of those graphs, generate a graph uniformly at random, and find an optimal one under some measurement, like the minimum weight.
However, the enumeration by those algorithms is labeled.
In other words, the obtained BDDs by those algorithms may have many isomorphic graphs.
Hence, the technique cannot be used, for example, for generating a graph at uniformly random when taking isomorphism into account.

This paper proposes polynomial-time algorithms for unlabeled intersection graph enumeration using BDDs.
The five intersection graph classes in concern are those of proper interval, cochain, bipartite permutation, (bipartite) chain, and threshold graphs. 
It is known that the unlabeled graphs with $n$ vertices of these classes have natural $\mathrm{O}(n)$-bit string encodings: We require $2n$ bits for proper interval and bipartite permutation graphs~\cite{SOYU12,SYKU10} and $n$ bits for chain, cochain, and threshold graphs~\cite{Mahadev1995threshold,PeledS95}. 
It may be a natural idea for enumerating those graphs to construct a BDD that represents those encoding strings. 
Here, we remark that there are different strings that represent isomorphic graphs, and we need to keep only a ``canonical'' one among those strings.
Actually, if we make a BDD naively represent those canonical strings, the resultant BDD will be exponentially large.
To solve the problem, we introduce new string encodings of intersection graphs of the respective classes so that the sizes of the BDDs representing canonical strings are polynomial in $n$.
Our encodings are still natural enough to extend the enumeration technique to more elaborate tasks:
namely, enumerating graphs with bounded maximum (bi)clique size and/or with maximum number of edges. 
One application of enumerating proper interval graphs with maximum clique size $k$ is, for example, to enumerate graphs with the bandwidth at most $k$.
Recall that the bandwidth of a graph is the minimum size of the maximum cliques in the proper interval graphs obtained by adding edges.
Thus, conversely, we can obtain graphs of bandwidth at most $k$ by removing edges from the enumerated graphs.

\section{Preliminary}
\noindent \textbf{Graphs. }
Let $G=(V, E)$ be a simple graph with $n$ vertices and $m$ edges. 
A sequence $P=(v_1, v_2, \dots, v_k)$ of vertices is a \emph{path from $v_1$ to $v_k$} if $v_i$ and $v_j$ are distinct for $i\neq j$ and $(v_i, v_{i+1})\in E$ for $i\in \{1, \dots, k-1\}$. 
The graph $G$ is \emph{connected} if for every two vertices $v_i, v_j\in V$, there exists a path from $v_i$ to $v_j$. 
The neighbor set of a vertex $v$ is denoted by $N(v)$, and the closed neighbor set of $v$ is denoted by $N[v] = N(v)\cup \{v\}$. 
A vertex $v$ is \emph{universal} if $|N(v)| = n-1$ and a vertex $v$ is \emph{isolate} if $|N(v)| = 0$. 
For $V'\subseteq V$ and $E'\subseteq E$ such that the endpoints of every edge in $E'$ are in $V'$, $G'=(V', E')$ is a \emph{subgraph} of $G$. 
The graph $G$ is \emph{complete} if every vertex is universal. If a subgraph $G'=(V', E')$ of $G$ is a complete graph, $V'$ is called a \emph{clique} of $G$. 
A clique $C$ is \emph{maximum} if for any clique $C'$ in $G$, $|C|\geq |C'|$. 
A vertex set $S$ is called an \emph{independent} set if for each $v\in S$, $N(v) \cap S = \emptyset$. 
The \emph{complement} of $G =(V,E)$ is the graph $\overline{G}=(V, \overline{E})$ where $\overline{E} = \{(u, v)\mid (u, v)\notin E\}$.

For a graph $G=(V, E)$, let $(X, Y)$ be a partition of $V$; that is, $V=X\cup Y$ and $X\cap Y=\emptyset$. 
A graph $G=(X\cup Y, E)$ is \emph{bipartite} if for every edge $(u, v)\in E$, either $u\in X$ and $v\in Y$ or $u\in Y$ and $v\in X$ holds. 
The bipartite graph $G$ is \emph{complete bipartite} if $E = \{(x, y)\mid x\in X, y\in Y\}$. 
For a subgraph $G'=(X'\cup Y', E')$ of $G$, $X'\cup Y'$ is called \emph{biclique} if $G'$ is complete bipartite. 
A biclique $B$ is \emph{maximum} if for any biclique $B'$ in $G$, $|B|\geq |B'|$. 
Note that we here say that a biclique has the ``maximum'' size if the number of not edges but vertices of it is maximum. 
For a bipartite graph $G=(X\cup Y, E)$, $\overline{G}$ is called \emph{cobipartite}. 
Note that $X$ and $Y$ are cliques in $\overline{G}$. 
An ordering $x_1, x_2, \dots, x_{|X|}$ on $X$ is an \emph{inclusion ordering} if $N(x_i)\cap Y \subseteq N(x_j)\cap Y$ for every $i, j$ with $i < j$. 

\smallskip
\noindent \textbf{Binary strings. }
We use the binary alphabet $\Sigma = \{\mtt{L}, \mtt{R}\}$ in this paper. 
Let $s = c_1 c_2\dots c_n$ be a binary string on $\Sigma^*$. 
The length of $s$ is $n$ and we denote it by $|s|$. 
Let $\overline{\mtt{L}} = \mtt{R}$ and $\overline{\mtt{R}} = \mtt{L}$. 
For a string $s = c_1 c_2 \dots c_n$, we define $\overline{s} = \overline{c_n}~\overline{c_{n-1}} \dots \overline{c_1}$. 
The \emph{height} $h_s(i)$ of $s$ at $i\in \{0, 1, \dots, n\}$ is defined by
$h_s(i) = |c_1 \dots c_i|_\mtt{L}-|c_1 \dots c_i|_\mtt{R}$, where $|t|_c$ denotes the number of occurrences of $c$ in a string $t$.
The string $s$ is \emph{balanced} if $h_s(n) = 0$; that is, the number of $\mtt{L}$ is equal to that of $\mtt{R}$ in $s$. 
The \emph{height} of $s$ is the maximum value in the height function for $s$ and denoted by $h(s)$; that is, $h(s) = \max_i{h_s(i)}$. 
We say $s$ is \emph{larger} than a string $s'$ with length $n$ if there exists an index $i\in \{1, \dots, n\}$ such that $h_s(i') = h_{s'}(i')$ for any $i' < i$ and $h_s(i) > h_{s'}(i)$, and we denote it by $s > s'$. 
The \emph{alternate} string $\alpha(s)$ of $s$ is obtained by reordering the characters of $s$ from outside to center, alternately; that is, $\alpha(s) = c_1 c_n c_2 c_{n-1} \dots c_{\lceil n/2 \rceil}$ if $n$ is odd and $\alpha(s) = c_1 c_n c_2 c_{n-1} \dots c_{n/2} c_{n/2+1}$ otherwise.

\smallskip
\noindent \textbf{Binary decision diagrams. }\label{sec:BDD}
A \emph{binary decision diagram (BDD)} is an edge labeled directed acyclic graph $D = (N,A)$ that classifies strings over a binary alphabet $\Sigma$ of a fixed length $n$.
To distinguish BDDs from the graphs we enumerate, we call elements of $N$ \emph{nodes} and those of $A$ \emph{arcs}.
The nodes are partitioned into $n+1$ groups: $N = N_1 \cup \dots \cup N_{n+1}$.
Nodes in $N_i$ are said to be at \emph{level $i$} for $1 \le i \le n+1$.
There is just one node at level $1$, called the \emph{root}.
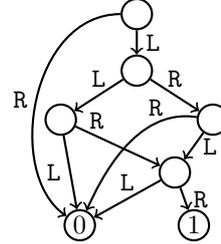
\begin{wrapfigure}{r}{0.3\textwidth}
\begin{tikzpicture}[xscale=0.5,yscale=-0.5,thick,shape=circle,inner sep=0pt,minimum size=4mm]
\newcommand{\point}[3][]{	\draw[#1] (#2,#3-0.2) -- (#2,#3+0.2);
}
\newcommand{\interval}[4][]{	\draw[#1] (#2,#4) -- (#3,#4);
}
\newcommand{\lrinterval}[4][]{	\interval[#1]{#2}{#3}{#4}
	\point[#1]{#2}{#4}
	\point[#1]{#3}{#4}
	\node () at (#2,#4-0.5) {$\mtt{L}$};
	\node () at (#3,#4-0.5) {$\mtt{R}$};
}
\node[draw] (0) at (0.5,5.6) {$0$};
\node[draw] (1) at (3.5,5.6) {$1$};
\node[draw] (3) at (3,4.2) { };
\node[draw] (2L) at (0,2.8) { };
\node[draw] (2R) at (4,2.8) { };
\node[draw] (1L) at (2,1.5) { };
\node[draw] (r) at (2,0) { };
\draw[->] (r) to node[right] {$\mtt{L}$} (1L);
\draw[->,out=180,in=225] (r) to node[left] {$\mtt{R}$} (0);
\draw[->] (1L) to node[above] {$\mtt{L}$} (2L);
\draw[->] (1L) to node[above] {$\mtt{R}$} (2R);
\draw[->] (2L) to node[left] {$\mtt{L}$} (0);
\draw[->] (2L) to node[near start, above] {$\mtt{R}$} (3);
\draw[->] (2R) to node[right] {$\mtt{L}$} (3);
\draw[->,out=195,in=285] (2R) to node[near start,above] {$\mtt{R}$} (0);
\draw[->] (3) to node[above] {$\mtt{L}$} (0);
\draw[->] (3) to node[right] {$\mtt{R}$} (1);
\end{tikzpicture}
\caption{An example BDD.\label{fig:bdd}}
\end{wrapfigure}
Level $(n+1)$ nodes are only two: the $0$-terminal node and the $1$-terminal node.
Each node in $N_i$ for $i \le n$ has two outgoing arcs pointing at nodes in $N_{i+1} \cup N_{n+1}$.
Thus, the length of every path from the root to a node in $N_i$ is just $i-1$ for $i \le n$.
The terminal nodes have no outgoing arcs.
The two arcs from a node have different labels from $\Sigma$. We call those arcs $\mtt{L}$-arc and $\mtt{R}$-arc. 
When a string $s = c_1 \dots c_n$ is given, we follow the arcs labeled $c_1,\dots,c_n$ from the root node.
If we reach the $1$-terminal, then the input is accepted.
If we reach the $0$-terminal, it is rejected.
One may reach a terminal node before reading the whole string.
In that case, we do not care the rest unread suffix of the string, and classify the whole string in accordance with the terminal node.
\figurename~\ref{fig:bdd} shows an example BDD, where $\mtt{LRLR}$ and $\mtt{LLRR}$ are accepted and $\mtt{LLRL}$ and $\mtt{RLRL}$ are rejected.

\section{Algorithms}

\subsection{Proper interval graphs and cochain graphs}\label{sec:pi}

\noindent \textbf{Definition and properties of proper interval graphs. }
A graph $G=(V, E)$ with $V=\{v_1, \dots, v_n\}$ is an \emph{interval} graph if there exists a set of $n$ intervals $\I=\{I_1, \dots, I_n\}$ such that $(v_i, v_j)\in E$ iff $I_i\cap I_j \neq \emptyset$ for $i, j\in \{1, \dots, n\}$. 
The set $\I$ of intervals is called an \emph{interval representation} of $G$. 
For an interval $I$, we denote the left and right endpoints of $I$ by $l(I)$ and $r(I)$, respectively. 
Without loss of generality, we assume that any two endpoints in $\I$ are distinct. 
An interval representation $\I$ is \emph{proper} if there are no two distinct intervals $I_i$ and $I_j$ in $\I$ such that $l(I_i) < l(I_j) < r(I_j) < r(I_i)$ or $l(I_j) < l(I_i) < r(I_i) < r(I_j)$. 
A graph $G$ is \emph{proper interval} if it has a proper interval representation. 

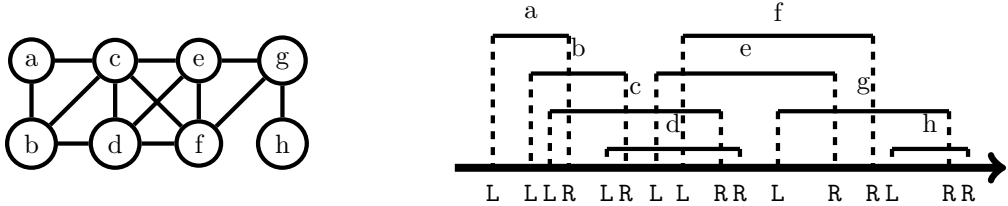
\begin{figure}[tbh]
 \centering
 \begin{minipage}{.35\textwidth}
  \centering
    \begin{tikzpicture}[scale=.55,circle,line width=1.6pt]
   \node[draw] (a) at (1, 2) {a};
   \node[draw] (b) at (1, 0) {b};
   \node[draw] (c) at (3, 2) {c};
   \node[draw] (d) at (3, 0) {d};
   \node[draw] (e) at (5, 2) {e};
   \node[draw] (f) at (5, 0) {f};
   \node[draw] (g) at (7, 2) {g};
   \node[draw] (h) at (7, 0) {h};

   \draw (a) -- (b);
   \draw (a) -- (c);
   \draw (b) -- (c);
   \draw (b) -- (d);
   \draw (c) -- (d);
   \draw (c) -- (e);
   \draw (c) -- (f);
   \draw (d) -- (e);
   \draw (d) -- (f);
   \draw (e) -- (f);
   \draw (e) -- (g);
   \draw (f) -- (g);
   \draw (g) -- (h);
  \end{tikzpicture}
 \end{minipage}
 \begin{minipage}{.64\textwidth}
  \centering
  \begin{tikzpicture}[scale=.50,circle,line width=1.6pt]
   \draw[->,line width=3pt] (0, 0.5) -- (14.5, 0.5);
   
   \draw (1, 4) -- node[midway, above] {a} (3, 4);
   \draw (2, 3) -- node[midway, above] {b} (4.5, 3);
   \draw (2.5, 2) -- node[midway, above] {c} (7, 2);
   \draw (4, 1) -- node[midway, above] {d} (7.5, 1);
   \draw (5.3, 3) -- node[midway, above] {e} (10, 3);
   \draw (6, 4) -- node[midway, above] {f} (11, 4);
   \draw (8.5, 2) -- node[midway, above] {g} (13, 2);
   \draw (11.5, 1) -- node[midway, above] {h} (13.5, 1);

   \begin{scope}[dashed]
   \draw (1, 4) -- (1, 0.5) node[below] {$\L$};
   \draw (2, 3) -- (2, 0.5) node[below] {$\L$};
   \draw (2.5, 2) -- (2.5, 0.5) node[below] {$\L$};
   \draw (4, 1) -- (4, 0.5) node[below] {$\L$};
   \draw (5.3, 3) -- (5.3, 0.5) node[below] {$\L$};
   \draw (6, 4) -- (6, 0.5) node[below] {$\L$};
   \draw (8.5, 2) -- (8.5, 0.5) node[below] {$\L$};
   \draw (11.5, 1) -- (11.5, 0.5) node[below] {$\L$};

   \draw (3, 4) -- (3, 0.5) node[below] {$\R$};
   \draw (4.5, 3) -- (4.5, 0.5) node[below] {$\R$};
   \draw (7, 2) -- (7, 0.5) node[below] {$\R$};
   \draw (7.5, 1) -- (7.5, 0.5) node[below] {$\R$};
   \draw (10, 3) -- (10, 0.5) node[below] {$\R$};
   \draw (11, 4) -- (11, 0.5) node[below] {$\R$};
   \draw (13, 2) -- (13, 0.5) node[below] {$\R$};
   \draw (13.5, 1) -- (13.5, 0.5) node[below] {$\R$};
   \end{scope}
  \end{tikzpicture}
 \end{minipage}
 \caption{Proper interval graph and its proper interval representation. The string representation of the proper interval representation is $\L\L\L\R\L\R\L\L\R\R\L\R\R\L\R\R$.}\label{fig:pig}
\end{figure}
Proper interval graphs can be represented by binary strings as follows. 
Let $G$ be a proper interval graph with $n$ vertices and $\I$ be a proper interval representation of $G$. 
We can represent $\I$ as a string by sweeping $\I$ from left to right and encoding $l(I)$ by $\mtt{L}$ and $r(I)$ by $\mtt{R}$, respectively. 
We denote the obtained string by $s(\I)$ and call it the \emph{string representation} of $\I$. 
The length of $s(\I)$ is $2n$. \begin{lemma}[\cite{SYKU10}]\label{lem:pig_string}
 Let $s(\I) = c_1 c_2 \dots c_{2n}$ be a string representation of a connected proper interval graph $G$ with $n$ vertices. 
 \begin{enumerate}
  \item $c_1 = \mtt{L}$ and $c_{2n} = \mtt{R}$, 
  \item $s(\I)$ is balanced; that is, the number of $\mtt{L}$ is same as that of $\mtt{R}$ in $s(\I)$, and
  \item $h_{s(\I)}(i) > 0$ for $i\in \{1, \dots, 2n-1\}$. 
 \end{enumerate}
\end{lemma}
A connected proper interval graph has at most two string representations~\cite{DHH96}. More strictly, 
 for any two string representations $s$ and $s'$ of a connected proper interval graph $G$, $s = s'$ or $s = \overline{s'}$. 
The string representation is said to be \emph{canonical} if $s>\overline{s}$ or $s=\overline{s}$. 
Thus, the canonical string representations have one-to-one correspondence to the proper interval graphs up to isomorphism~\cite{SYKU10}. 

\smallskip
\noindent \textbf{Algorithm for $n$ vertices. }
We here present an enumeration algorithm of all connected proper interval graphs with $n$ vertices up to isomorphism. 
We would like to construct a BDD representing all canonical string representations of proper interval graphs. 
However, for the efficiency of the BDD construction as described later, we instead construct a BDD representing alternate strings of all canonical representations of proper interval graphs. 

We describe an overview of our algorithm. We construct the BDD in a breadth-first manner
in the direction from the root node to the terminals. We create the root node in $N_1$, and for each node in $N_i \ (i\in \{1,\ldots,2n\})$, we create its $\mtt{L}$ and $\mtt{R}$-arcs
and make each arc point at one of the existing nodes in $N_{i+1}$ or $N_{2n+1}$ or a newly created node. We call making an arc point at 0-terminal node \emph{pruning}.
For each node $\nu$, we store into $\nu$ information on the paths from the root to $\nu$ as a tuple, which we call \emph{state}.
Two nodes having the same state never exist.
When creating an ($\mtt{L}$ or $\mtt{R}$) arc of a node,
we compute the state of the destination from the state of the original node.
If there is an existing node having the same state as the computed one,
we make the arc point at the existing node,
which we call \emph{(node) sharing}.

Consider deciding whether a string $s$ in $\Sigma^{2n}$ is canonical or not;
that is, $s>\overline{s}$ or $s=\overline{s}$ holds.
Suppose that $s = c_1 c_2 \dots c_{2n}$ and we have $\overline{s} = \overline{c_{2n}}\,\overline{c_{2n-1}} \dots \overline{c_1}$.
This can be done by comparing $c_i$ with $\overline{c_{2n - i + 1}}$ for $i = 1,\ldots,2n$.
When creating a node $\nu$ in the BDD construction process,
we would like to conduct pruning early if we can determine that all the path labels from the
root via $\nu$ will not be canonical.
That is the reason we adopt alternate string representations.
A node in level $i\ (\in \{1,\ldots,2n\})$ corresponds to the $\lceil i/2\rceil$th character in the string representation if $i$ is odd,
and the $(2n+1 - i/2)$th one otherwise.
For example, consider the path $\mtt{L}\mtt{R}\mtt{L}\mtt{R}\mtt{R}\mtt{R}$.
Any path extending $\mtt{LRLRRR}$ will represent a string of the form $s = \mtt{LLR} t \mtt{RRR}$ for which $\overline{s} = \mtt{LLL}\overline{t}\mtt{LRR}$ for some $t \in \Sigma^{*}$ and $s < \overline{s}$ holds.
This implies $s$ cannot be canonical. 
The path goes to the 0-terminal.

We make each node, say $\nu$, maintain state $(i, h_{\mathrm{L}}, h_{\mathrm{R}}, F)$.
The first element $i$ is the level where $\nu$ is.
We take an arbitrary path from the root node to $\nu$,
say $c_1 c_{2n} c_2 c_{2n-1} \dots c_{\lceil i/2 \rceil - 1} c_{2n+2-\lceil i/2 \rceil}$ (the case where $i$ is odd)
or $c_1 c_{2n} c_2 c_{2n-1} \dots \linebreak[1] c_{2n+2-\lceil i/2 \rceil} c_{\lceil i/2 \rceil}$ (the case where $i$ is even).
The second and third elements $h_{\mathrm{L}}, h_{\mathrm{R}}$ represent the heights of the sequences $c_1 c_2 \dots c_{\lfloor i/2 \rfloor}$
and $\overline{c_{2n}}\,\overline{c_{2n-1}} \dots \overline{c_{2n+2-\lceil i/2 \rceil}}$, respectively.
Note that we must design an algorithm so that it is well-defined; that is, the values of the sequences obtained from all the paths from the root node to $\nu$ are the same.
$F$ represents whether ($\star$) $c_{\hat{\imath}} = \overline{c_{2n+1-\hat{\imath}}}$ holds for all $\hat{\imath} = 1,\ldots, \lceil i/2 \rceil - 1$.
If $F = \top$, $(\star)$ does not hold; that is, there exists $i'$ such that $c_{i'} \neq \overline{c_{2n+1-i'}}$.
If $c_{i'} = \mtt{R}$ and $\overline{c_{2n+1-i'}} = \mtt{L}$, the canonicity condition does not meet.
As shown later, such a node never exists because we conduct the pruning.
Therefore, $F = \top$ means that $c_{i'} = \mtt{L}$, $\overline{c_{2n+1-i'}} = \mtt{R}$ and
$c_{i''} = \overline{c_{2n+1-i''}}$ holds for all $i'' \le i' - 1$, which implies that the canonicity condition
is satisfied whatever the other characters are. $F = \bot$ means that $(\star)$ holds.

We discuss how to store states and conduct pruning in the process of the BDD construction.
We make the root node have the state $(1, 0, 0, \bot)$.
Let $\nu$ be a node that has the state $(i, h_{\mathrm{L}}, h_{\mathrm{R}}, F)$ and $\nu_{\mtt{L}}$ and $\nu_{\mtt{R}}$ be nodes pointed at by $\mtt{L}$-arc and $\mtt{R}$-arc of $\nu$.
If $i = 1$, $\nu_{\mtt{R}}$ is 0-terminal, and if $i = 2$, $\nu_{\mtt{L}}$ is 0-terminal because of the condition (i) in Lemma~\ref{lem:pig_string}.
First, we consider the case where $i$ is odd. $\mtt{L}$-arc and $\mtt{R}$-arc of $\nu$ mean that the $\lceil i/2 \rceil$th character is $\mtt{L}$ and $\mtt{R}$, respectively.
We make $\nu_{\mtt{L}}$ have state $(i + 1, h_{\mathrm{L}}+1, h_{\mathrm{R}}, F)$.
As for $\mtt{R}$-arc, if $h_{\mathrm{L}}-1 \leq 0$, we make $\mtt{R}$-arc of $\nu$ point at 0-terminal because $\mtt{R}$-arc means $c_{\lceil i/2 \rceil} = \mtt{R}$ and the height of $c_1 c_2 \dots c_{\lfloor (i+1)/2 \rfloor}$ violates the condition of (iii) in Lemma~\ref{lem:pig_string}.
Otherwise, we make $\nu_{\mtt{R}}$ have state $(i + 1, h_{\mathrm{L}}-1, h_{\mathrm{R}}, F)$.
Next, we consider the case where $i$ is even.
$\mtt{L}$-arc and $\mtt{R}$-arc of $\nu$ mean that the $(2n + 1 - \lceil i/2 \rceil)$th character is $\mtt{L}$ and $\mtt{R}$, respectively.
If $F = \top$, we make $\nu_{\mtt{L}}$ and $\nu_{\mtt{R}}$ maintain states $(i + 1, h_{\mathrm{L}}, h_{\mathrm{R}}-1, \top)$ and $(i + 1, h_{\mathrm{L}}, h_{\mathrm{R}}+1, \top)$, respectively.
(Recall that since $F = \top$ means that the canonicity condition has already been satisfied,
we need not update $F$.)
We conduct pruning for $\mtt{L}$-arc if $h_{\mathrm{R}}-1 \le 0$.
Let us consider the case where $F = \bot$.
Recall that ($\star$) holds. Although we want to compare the $\lceil i/2 \rceil$th and
$(2n + 1 - \lceil i/2 \rceil)$th characters to decide whether the canonicity condition holds or not,
$\nu$ does not have the information on the $\lceil i/2 \rceil$th character.
Instead, $\nu$ has $h_{\mathrm{L}}$ and $h_{\mathrm{R}}$. We consider two cases (i) and (ii):
(i) If $h_{\mathrm{L}} - 1 = h_{\mathrm{R}}$, it means that the $\lceil i/2 \rceil$th character is $\mtt{L}$.
In this case, $\mtt{R}$-arc of $\nu$ means that the $(2n + 1 - \lceil i/2 \rceil)$th character is $\mtt{R}$,
which implies that ($\star$) still holds. Therefore, we make $\nu_{\mtt{R}}$ maintain state $(i + 1, h_{\mathrm{L}}, h_{\mathrm{R}}+1, \bot)$.
$\mtt{L}$-arc of $\nu$ means that the $(2n + 1 - \lceil i/2 \rceil)$th character is $\mtt{L}$,
which implies that ($\star$) no longer holds and the canonicity condition is satisfied. Therefore, we make $\nu_{\mtt{L}}$ maintain state $(i + 1, h_{\mathrm{L}}, h_{\mathrm{R}}-1, \top)$.
(ii) If $h_{\mathrm{L}} - 1 \neq h_{\mathrm{R}}$, it means that the $\lceil i/2 \rceil$th character is $\mtt{R}$.
In this case, $\mtt{R}$-arc of $\nu$ means that the $(2n + 1 - \lceil i/2 \rceil)$th character is $\mtt{R}$,
which violates the canonicity condition. We make $\mtt{R}$-arc of $\nu$ point at 0-terminal.
$\mtt{L}$-arc of $\nu$ means that the $(2n + 1 - \lceil i/2 \rceil)$th character is $\mtt{L}$,
which implies that ($\star$) still holds. Therefore, we make $\nu_{\mtt{L}}$ maintain state $(i + 1, h_{\mathrm{L}}, h_{\mathrm{R}}-1, \bot)$.

Consider the case where $i = 2n$ (final level).
Let the computed state as the destination of $\mtt{L}$- or $\mtt{R}$-arc of a node in $N_{2n}$ be $(2n + 1, h'_{\mathrm{L}}, h'_{\mathrm{R}}, F')$.
If $h'_{\mathrm{L}} \neq h'_{\mathrm{R}}$, the destination is pruned (0-terminal)
because it violates the condition of (ii) in Lemma~\ref{lem:pig_string}.
Otherwise, we make the arc point at 1-terminal.

\begin{theorem}
 Our algorithm constructs a BDD representing all canonical string representations of connected proper interval graphs in $\mathrm{O}(n^3)$ time and space. 
\end{theorem}
\begin{proof}
We here analyze the complexity of the algorithm. 
For each level $i\in\{0, 1, \dots, 2n\}$, the number of nodes in $N_i$ is $\mathrm{O}(n^2)$ because $0\leq h_{\mathrm{L}}, h_{\mathrm{R}} \leq n$ and $F\in \{\bot, \top\}$. 
Thus, the total size of BDD is $\mathrm{O}(n^3)$. 
The computation of the next state for each node can be run in constant time because it has only increment and we can access the nodes in constant time by using $\mathrm{O}(n^2)$ pointers. 
\end{proof}

\smallskip
\noindent \textbf{Algorithm for maximum clique size $k$. }
We here present an algorithm that given natural numbers $n$ and $k$, enumerates all proper interval graphs with $n$ vertices and the maximum clique size at most $k$. 
It is well known that a clique of an interval graph $G$ corresponds to overlap intervals of a point in an interval representation of $G$~\cite{CLRS09}. 
The number of overlapping intervals is same as the height of string representation of a proper interval graph. 
Thus, the enumeration of all proper interval graphs with the maximum clique size at most $k$ can be seen as that of all canonical string representations with the height at most $k$. 
We modify the algorithm for $n$ vertices by adding one pruning for the case when either of the heights $h_\mrm{L}$ or $h_\mrm{R}$ becomes larger than $k$.
Therefore, our extended algorithm runs in $\mathrm{O}(k^2n)$ time and space since the ranges of $h_{\mathrm{L}}$ and $h_{\mathrm{R}}$ become $k$ from $n$.

\smallskip \noindent \textbf{Algorithm for $m$ edges. }
To extend the algorithm for $n$ vertices and $m$ edges, we here show how to count the number of edges from the string representation. 
Let $s$ be a string representation of a proper interval graph with $m$ edges. 
Sweeping the string representation from left to right, for each $i\in \{1, \dots, 2n\}$ with $c_i=\mtt{L}$, the height $h_s(i)$ is the number of intervals $I_j$ with $j <i$ that overlap with $i$. This means that the vertex $v$ corresponding to $c_i$ is incident to $h_s(i)$ edges in $G$. 
Thus, we obtain the number of edges from the string representation as follows. 
\begin{lemma}\label{lem:PIGE}
 Let $s=c_1 \dots c_{2n}$ be a string representation of a connected proper interval graph $G$ with $m$ edges and $J$ be the set of indices $i$ of $s$ such that $c_i = \mtt{L}$. 
 The summation of heights in $J$ is equal to $m$; i.e., $\sum_{i\in J}{h_s(i)} = m$. 
 \end{lemma}

In the construction of a BDD, each node stores the value to maintain the number of edges $m'$.
The state of each node is now a quintuple $(i, h_{\mathrm{L}}, h_{\mathrm{R}}, F, m')$. 
For the $\mtt{L}$-arc of a node $\nu$, the number of edges $m'$ is updated to $m'+h_{\mathrm{L}}$ if $i$ is odd and to $m'+h_{\mathrm{R}}-1$ otherwise. 
When either $i$ is odd and $m'+h_{\mathrm{L}} > m$ or $i$ is even and $m'+h_{\mathrm{R}}-1 > m$ holds, we make the $\mtt{L}$-arc of $\nu$ point at the 0-terminal since the number of edges is larger than $m$. 
We make each arc point at the $1$-terminal if it gives a state $(2n+1,h,h,F,m)$ for some $h$ and $F$ based on the state updating rule.
Otherwise, it must point at the $0$-terminal.
For each $i\in\{1, \dots, 2n\}$, 
the number of nodes in $N_i$ is $\mathrm{O}(n^2m)$ since $0\leq h_{\mathrm{L}}, h_{\mathrm{R}} \leq n$ and $0\leq m'\leq m$ and the number of levels is $2n$. 
Therefore, the algorithm runs in $\mathrm{O}(n^3m)$ time. \begin{theorem}
 A BDD representing all connected proper interval graphs with $n$ vertices and maximum clique size $k$ and with $n$ vertices and $m$ edges can be constructed in $\mathrm{O}(k^2n)$ time and $\mathrm{O}(n^3m)$ time, respectively. 
\end{theorem}

\smallskip \noindent \textbf{Cochain graphs. }
A graph $G=(X\cup Y, E)$ is a \emph{cochain} graph if $G$ is cobipartite and each of $X$ and $Y$ has an inclusion ordering. In other words, $X$ and $Y$ are cliques in $G$ and we have two orderings over $X=\{x_1,\dots,x_{n_X}\}$ and $Y=\{y_1,\dots,y_{n_Y}\}$ such that $(x_i,y_j) \in E$ implies $(x_{i'},y_{j'}) \in E$ for any $i \le i'$ and $j \le j'$.
It is well-known~\cite{Brandstadt:1999} that cochain graphs are a subclass of proper interval graphs.
Here, we give a concrete proper interval representation $\{I_1,\dots,I_{n_X},J_1,\dots,J_{n_Y}\}$ of $G$, where $x_i$ and $y_j$ correspond to $I_i$ and $J_j$, respectively, by
\begin{itemize}
	\item $l(I_1) < \dots < l(I_{n_X}) < r(I_1) < \dots < r(I_{n_X}) < r(J_{n_Y})$,
	\item $l(I_{n_X}) < l(J_{n_Y}) < \dots < l(J_{1}) < r(J_{n_Y}) < \dots < r(J_{1})$,
	\item $l(J_j) < r(I_i) $ iff $(x_i,y_j) \in E$ for $1 \le i \le n_X$ and $1 \le j \le n_Y$.
\end{itemize}
The inclusion ordering constraint guarantees that the above is well-defined and gives a proper interval representation.
Therefore, one can specify a cochain graph as a proper interval graph by a $2n$-bit string representation.
Moreover, the strong restriction of cochain graphs allows us to reduce the number of bits to specify a cochain graph.
Obviously, the first $n_X$ bits of the proper interval string representation of a cochain graph are all $\mtt{L}$ and the last $n_Y$ bits are all $\mtt{R}$.
Thus, those $n=n_X+n_Y$ bits are redundant and removable.
Indeed, one can recover the numbers $n_X$ and $n_Y$ from the remaining $n$ bits.
Since every surviving bit of $\mtt{R}$ corresponds to $r(I_i)$ for some $i$, the number of those bits is just $n_X$.
Similarly, $n_Y$ is the number of bits of $\mtt{L}$ in the new $n$-bit representation.
Conversely, every $n$-bit string $s$ can be seen as the string representation of a cochain graph with $n$ vertices.
However, the $n$-bit strings are not in one-to-one correspondence to the cochain graphs because universal vertices in the cochain graphs can be seen in either $X$ or $Y$. 
To avoid the duplication, we assume that all universal vertices are in $Y$, so we only consider $n$-bit strings without $\mtt{R}$ as a suffix. 
Using this $n$-bit string representation, we obtain an enumeration algorithm for cochain graph, and it runs in $\mathrm{O}(n)$ time. 

For the constraint problems, we use $2n$-bit strings because we need to compute the size of cliques or the number of edges. 
Our algorithms with constraints for cochain graphs are similar to that of proper interval graphs and need to recognize whether the strings represent cochain graphs. 

\begin{theorem}
 A BDD representing all canonical string representations of cochain graphs with $n$ vertices, $n$ vertices and maximum clique size $k$, and $n$ vertices and $m$ edges can be constructed in $\mathrm{O}(n)$, $\mathrm{O}(k^2n)$, and $\mathrm{O}(n^3m)$ time, respectively. 
\end{theorem}

\subsection{Bipartite permutation graphs and chain graphs}
\noindent \textbf{Definition and properties of bipartite permutation graphs. }
Let $\pi$ be a permutation on $V$; that is,
$\pi$ is a bijection from $V$ to $\{1, \dots, n\}$.
We define $\overline{\pi}$ as $\overline{\pi}(v) = n + 1 - \pi(v)$ for all $v \in V$.
We denote by $\pi^{-1}$ the inverse of $\pi$.

A graph $G=(V, E)$ is \emph{permutation} if it has
a pair $(\pi_1, \pi_2)$ of two permutations on $V$ such that there exists an edge $(u, v)\in E$ iff
$(\pi_1(u) - \pi_1(v)) (\pi_2(u) - \pi_2(v)) < 0$.
The pair $\P = (\pi_1, \pi_2)$ can be seen as the following intersection model
on two parallel horizontal lines $L_1$ and $L_2$:
the vertices in $V$ are arranged on the line $L_1$ (resp.\ line $L_2$)
according to $\pi_1$ (resp.\ $\pi_2$).
Each vertex $w$ corresponds to a line segment $l_w$, which joins $w$
on $L_1$ and $w$ on $L_2$.
An edge $(u, v)$ is in $E$ iff $l_u$ and $l_v$ intersects,
which is equivalent to $(\pi_1(u) - \pi_1(v)) (\pi_2(u) - \pi_2(v)) < 0$.
The model $\P=(\pi_1, \pi_2)$ is called a \emph{permutation diagram}.
A graph $G$ is \emph{bipartite permutation} if $G$ is bipartite and permutation.

Let $\P=(\pi_1, \pi_2)$ be a permutation diagram of a connected bipartite permutation graph $G=(V, E)$.
Let us observe properties of $\pi_1$ and $\pi_2$, which are discussed in~\cite{SOYU12}.
First, there is no vertex $u \in V$ such that $\pi_1(u) = \pi_2(u)$ unless $n = 1$.
Secondly, for all vertices $u, v \in V$ such that
$\pi_1(u) < \pi_2(u)$, $\pi_1(v) < \pi_2(v)$ and $\pi_1(u) < \pi_1(v)$ hold,
$\pi_2(u) > \pi_2(v)$ does not hold; that is, $l_u$ and $l_v$ never intersects.
Therefore, $X = \{u \mid \pi_1(u) < \pi_2(u) \}$
and $Y = \{u \mid \pi_1(u) > \pi_2(u) \}$ give the vertex partition of $G$.
By expressing the above observation with the intersection model,
the line segments are never straight vertical and classified into $X$ and $Y$ depending on their tilt directions:
lines in $X$ go from upper left to lower right and those in $Y$ go from lower left to upper right.

Based on the above discussion, let us give a string representation $s(\P)$ of the permutation diagram $\P$.
We define $s_x(\P) = x_1 \dots x_n$ and $s_y(\P) = y_1 \dots y_n$ as follows:
For $i = 1,\ldots,n$, $x_i = \L$ if $\pi_1(\pi_1^{-1}(i)) (= i) < \pi_2(\pi_1^{-1}(i))$,
and $x_i = \R$ otherwise.
Similarly, for $i = 1,\ldots,n$, $y_i = \R$ if $\pi_2(\pi_2^{-1}(i)) (= i) > \pi_1(\pi_2^{-1}(i))$, and $y_i = \L$ otherwise.
In other words, $x_i = \L$ iff the $i$th intersection point of $L_1$ is with a line segment from $X$ in the intersection model.
On the other hand, $y_i = \L$ iff the $i$th intersection point of $L_2$ is with a line segment from $Y$.
We define the string representation $s(\P)$ of $\P$ by
$s(\P) = x_1y_1x_2y_2 \dots x_n y_n$.
The string representation $s(\P)$ has the following properties~\cite{SOYU12}.
\begin{lemma}\label{lem:bpg_string}
 Let $s = c_1 c_2 \dots c_{2n}$ be a string representation of a connected bipartite permutation graph $G$ with $n$ vertices. Then,
 \begin{enumerate}
  \item[(i)] $c_1 = \mtt{L}$ and $c_{2n} = \mtt{R}$,
  \item[(ii)] $s$ is balanced; that is, the number of $\mtt{L}$ is the same as that of $\mtt{R}$ in $s$, and
  \item[(iii)] $h_{s}(i) > 0$ for $i\in \{1, \dots, 2n-1\}$.
 \end{enumerate}
\end{lemma}
By horizontally, vertically, and rotationally flipping $\P$,
we obtain essentially equivalent diagrams $\P^{\mathrm{V}} = (\pi_2, \pi_1)$, $\P^{\mathrm{H}} = (\overline{\pi_1}, \overline{\pi_2})$, and $\P^{\mathrm{R}} = (\overline{\pi_2}, \overline{\pi_1})$ of $G$, respectively.
\begin{lemma}[\cite{SOYU12}]\label{lem:bpg_canonical}
 Let $\P_1$ and $\P_2$ be permutation diagrams of a connected bipartite permutation graph.
 At least one of the equations $s(\P_1) = s(\P_2)$, $s(\P_1) = s(\P_2^{\mathrm{V}})$, $s(\P_1) = s(\P_2^{\mathrm{H}})$, or $s(\P_1) = s(\P_2^{\mathrm{R}})$ holds.
\end{lemma}
A string representation $s(\P)$ is said to be \emph{canonical} if all the inequalities $s(\P)\geq s(\P^{\mathrm{V}})$, $s(\P)\geq s(\P^{\mathrm{H}})$, and $s(\P)\geq s(\P^{\mathrm{R}})$ hold. 

\smallskip \noindent \textbf{Algorithm for $n$ vertices. }
We construct the BDD representing the set of bipartite permutation graphs using the alternate strings of the canonical representation strings.
Each BDD node is identified with a state tuple $(i, h_{\mathrm{L}}, h_{\mathrm{R}}, c_{\mathrm{L}}, c_{\mathrm{R}},$\\ $F_{\mathrm{V}}, F_{\mathrm{H}}, F_{\mathrm{R}})$.
The integer $i$ is the level where the node is.
The heights $h_{\mathrm{L}}$ and $h_{\mathrm{R}}$ are
those of $x_1x_2 \dots x_n$ and $\overline{y_n}\,\overline{y_{n-1}} \dots \overline{y_1}$, respectively,
the purpose of which is the same as in Sec.~\ref{sec:pi}.

Let us describe $F_{\mathrm{V}}, F_{\mathrm{H}}$ and $F_{\mathrm{R}}$.
$F_{\mathrm{R}}$ is $\bot$ or $\top$,
which is used for deciding whether $s(\P)\geq s(\P^{\mathrm{R}})$ holds or not.
Recall that if $s(\P) = x_1y_1x_2y_2 \dots x_n y_n$,
$s(\P^{\mathrm{R}}) = \overline{y_n}\,\overline{x_n}\,\overline{y_{n-1}}\,\overline{x_{n-1}} \dots \overline{y_1}\,\overline{x_1}$.
According to the variable order $\alpha(s(\P))$,
we can decide whether $s(\P)> s(\P^{\mathrm{R}})$ holds or not using the heights $h_{\mathrm{L}}$ and $h_{\mathrm{R}}$
by the way described in Sec.~\ref{sec:pi}.
Then, $F_{\mathrm{R}}$ has the same role as $F$ in Sec.~\ref{sec:pi}.
Next, we consider $F_{\mathrm{V}}$, which is used for
deciding the canonicity of $s(\P)\geq s(\P^{\mathrm{V}})$.
Recall that if $s(\P) = x_1y_1x_2y_2 \dots x_n y_n$,
$s(\P^{\mathrm{V}}) = y_1x_1y_2x_2 \dots y_n x_n$.
We need to compare $x_1$ with $y_1$, $y_1$ with $x_1,\ldots,$ and $y_n$ with $x_n$
in order.
Recall that on the BDD, the value of $y_i$ is represented by arcs of each node in level $4 i - 1$.
The value of $x_i$ has already been determined by arcs of a node in level $4 i - 3$.
Therefore, to compare $x_i$ with $y_i$, we store the value of $x_i$ into nodes.
Strictly speaking,
if $i$ is odd, 
then, $c_{\mathrm{L}} = x_{\lceil i / 2 \rceil - 1}$ and $c_{\mathrm{R}} = y_{2n - \lceil i / 2 \rceil + 2}$.
If $i$ is even, 
then, $c_{\mathrm{L}} = x_{i/2}$ and $c_{\mathrm{R}} = y_{2n - i/2 + 2}$.
The stored values $c_{\mathrm{L}}$ and $c_{\mathrm{R}}$ are also used for deciding
whether $s(\P)\geq s(\P^{\mathrm{H}})$ holds or not in a similar way.

We estimate the number of BDD nodes by
counting the possible values of a state
$(i, h_{\mathrm{L}}, h_{\mathrm{R}}, c_{\mathrm{L}}, c_{\mathrm{R}},\allowbreak
F_{\mathrm{V}}, F_{\mathrm{H}}, F_{\mathrm{R}})$.
Since $1 \le i \le 2n$, $0 \le h_{\mathrm{L}} \le n$, $0 \le h_{\mathrm{R}} \le n$,
and the number of possible states of
$c_{\mathrm{L}}, c_{\mathrm{R}}, F_{\mathrm{V}}, F_{\mathrm{H}}, F_{\mathrm{R}}$
are two, the number of possible values of tuples is
$2n \times (n + 1)^2 \times 2^5
= \mathrm{O}(n^3)$.

\smallskip \noindent \textbf{Algorithm for $m$ edges. }
We present an algorithm that constructs the BDD representing
the set of (string representations of) bipartite permutation graphs
with $n$ vertices and $m$ edges when $n$ and $m$ are given.
The number of edges of a bipartite permutation graph $G$ is
that of intersections of the permutation diagram of $G$.
We use the following lemma.
\begin{lemma}
    The number of edges is $\sum_{i=1}^{n} h_{s(\P)}(2i)$.
\end{lemma}
We can easily obtain
\begin{equation}\label{eq:hvalue}
    \sum_{i=1}^{n} h_{s(\P)}(2i) = \sum_{i=1}^{\lceil n/2 \rceil} h_{s(\P)}(2i) + \sum_{i=1}^{\lfloor n/2 \rfloor} h_{s(\P^{\mathrm{R}})}(2i).
\end{equation}
To count the number of edges, we store this value into each BDD node.
Let us describe the detail. 
We make each BDD node maintain a tuple $(i, h_{\mathrm{L}}, h_{\mathrm{R}}, c_{\mathrm{L}}, c_{\mathrm{R}},$\\ $F_{\mathrm{V}}, F_{\mathrm{H}}, F_{\mathrm{R}}, m')$.
The first eight elements are the same as the ones described above.
The last element $m'$ is the current value of (\ref{eq:hvalue}).
Thus, the running time of the algorithm is $\mathrm{O}(n^3 m)$.
\begin{theorem}
 A BDD representing all connected bipartite permutation graphs with $n$ vertices, and $n$ vertices and $m$ edges can be constructed in $\mathrm{O}(n^3)$ and $\mathrm{O}(n^3m)$ time, respectively. 
\end{theorem}

\smallskip \noindent \textbf{Chain graphs. }
A graph $G=(X\cup Y, E)$ is a \emph{chain} graph if $G$ is bipartite and each of $X$ and $Y$ has an inclusion ordering. Let $(x_1, x_2, \dots, x_{|X|})$ and $(y_1, y_2, \dots, y_{|Y|})$ be an inclusion ordering of $X$ and $Y$, respectively. 
Chain graphs are known to be a subclass of bipartite permutation graphs~\cite{Brandstadt:1999} and have the following permutation diagrams $\P=(\pi_1, \pi_2)$~\cite{OkamotoUU09}:
\begin{itemize}
 \item $\pi_1 = (x_1, x_2, \dots, x_{|X|}, y_{|Y|}, y_{|Y|-1}, \dots, y_1)$,  \item for $i, j\in\{1, \dots, |X|\}$ with $i<j$, $\pi_2(x_i) < \pi_2(x_j)$,
 \item for $i, j\in\{1, \dots, |Y|\}$ with $i<j$, $\pi_2(y_j) < \pi_2(y_i)$.
\end{itemize}
Chain graphs as bipartite permutation graphs have $2n$-bit string representations based on the permutation diagrams. From the diagram and Lemma~\ref{lem:bpg_canonical}, we observe that the string of $\pi_1$ is uniquely determined except for exchanging $X$ and $Y$. 
Since $\pi_1$ can be fixed as above, any chain graph can be represented using an $n$-bit string by sweeping $\pi_2$: The $i$th element of $\pi_2$ is encoded as $\L$ if $\pi_2^{-1}(i)\in X$ and is encoded as $\R$ if $\pi_2^{-1}(i)\in Y$. 
If a chain graph $G$ is disconnected, $G$ consists of two parts: a connected chain graph component and a set of isolated vertices~\cite{KOSU12}. 
We observe that the connected chain graphs have a one-to-one correspondence with the string representations up to reversal~\cite{PeledS95}. 
On the other hand, isolated vertices may arbitrarily belong to $X$ or $Y$.
To determine a unique string representation,
we assume that isolated vertices are all in $X$, where the representation strings must not end with $\mtt{R}$.
Thus, we obtain an algorithm to construct a BDD representing all canonical $n$-bit string representations of chain graphs and it runs in $\mathrm{O}(n)$. 
For the restriction problems, we adopt $2n$-bit strings defined as representations of bipartite permutation graphs instead of $n$-bit representations to compute the number of edges or the size of bicliques. 
In the algorithms, we need to check whether the constructed strings represent chain graphs satisfying the conditions described above. 
\begin{theorem}
 A BDD representing all chain graphs with $n$ vertices, $n$ vertices and maximum biclique size $k$, and $n$ vertices and $m$ edges can be constructed in $\mathrm{O}(n)$, $\mathrm{O}(k^2n)$, $\mathrm{O}(n^3m)$ time, respectively.  
  \end{theorem}

\subsection{Threshold graphs}
A graph $G$ is a \emph{threshold} graph if the vertex set of $G$ can be partitioned into $X$ and $Y$ such that $X$ is a clique and $Y$ is an independent set and each of $X$ and $Y$ has an inclusion ordering. 
Threshold graphs are a subclass of interval graphs, and any threshold graph can be constructed by the following process~\cite{Brandstadt:1999,Mahadev1995threshold}. 
First, if the size of the vertex set is one, the graph is threshold. 
Then, for a threshold graph $G$, (1) the graph by adding an isolated vertex to $G$ is also threshold, and (2) the graph adding a universal vertex to $G$ is also threshold. 
The sequence of the two operations (1) and (2) to construct a threshold graph is called a \emph{construction sequence}. 
It is easy to see that the two threshold graphs $G_1$ and $G_2$ are not isomorphic if the construction sequences of (1) and (2) of $G_1$ and $G_2$ are different.
From this characterization of threshold graphs, we obtain algorithms to construct a BDD representing all unlabeled threshold graphs by encoding the construction sequences of the operation (1) to $\L$ and (2) to $\R$. 
\begin{theorem}
 A BDD representing all threshold graphs with $n$ vertices, $n$ vertices and maximum clique size $k$, and $n$ vertices and $m$ edges can be constructed in $\mathrm{O}(n)$ time, $\mathrm{O}(kn)$ time, and $\mathrm{O}(nm)$ time, respectively. 
\end{theorem}

\subsubsection*{Acknowledgments.}
The authors are grateful for the helpful discussions of this work with Ryuhei Uehara. 
This work was supported in part by JSPS KAKENHI Grant Numbers JP18H04091, JP19K12098, JP20H05794, and JP21H05857.

 \bibliographystyle{plain}
 \bibliography{ref}

\end{document}